\documentclass[conference]{IEEEtran}

\makeatletter
\def\ps@headings{%
\def\@oddhead{\mbox{}\scriptsize\rightmark \hfil \thepage}%
\def\@evenhead{\scriptsize\thepage \hfil \leftmark\mbox{}}%
\def\@oddfoot{}%
\def\@evenfoot{}}
\makeatother
\usepackage{amstext,bbm,tabls}
\usepackage{amssymb,amsthm}
\usepackage{amsfonts}
\usepackage{latexsym}
\usepackage[mathscr]{euscript}
\usepackage{amsmath}

\usepackage{graphics}
\usepackage{graphicx}
\usepackage{verbatim}
\usepackage{cite}
\usepackage{comment}
\usepackage{algorithm}
\usepackage{algorithmic}
\newcommand{\beq}{\begin{equation}}
\newcommand{\eeq}{\end{equation}}
\newcommand{\beqq}{\begin{equation*}}
\newcommand{\eeqq}{\end{equation*}}
\newtheorem{definition}{Definition}

\newtheorem{proposition}{Proposition}

\theoremstyle{remark}

\newcommand{\bi}{\begin{itemize}}
\newcommand{\ei}{\end{itemize}}
\newcommand{\E}{\mathbbm{E}}

\newcommand{\ls}[1]
 {\dimen0=\fontdimen6\the\font \lineskip=#1\dimen0
\advance\lineskip.5\fontdimen5\the\font \advance\lineskip-\dimen0
\lineskiplimit=.9\lineskip \baselineskip=\lineskip
\advance\baselineskip\dimen0 \normallineskip\lineskip
\normallineskiplimit\lineskiplimit \normalbaselineskip\baselineskip
\ignorespaces }
\ls{1.0}
\begin{document}
\title{A Game Theoretic Model for \\Network Virus Protection}
 \author{\IEEEauthorblockN{Iyed Khammassi\IEEEauthorrefmark{1}\IEEEauthorrefmark{2}, Rachid Elazouzi\IEEEauthorrefmark{1}, Majed Haddad\IEEEauthorrefmark{1} and Issam Mabrouki\IEEEauthorrefmark{2} }
\IEEEauthorblockA{\IEEEauthorrefmark{1}University of Avignon, 84000 Avignon, FRANCE\\
Email: firstname.lastname@univ-avignon.fr}
\IEEEauthorblockA{\IEEEauthorrefmark{2}University of Manouba Manouba, Tunisia\\
issam.mabrouki@hanalab.org}
}

\date{\today}
\maketitle
\begin{abstract}
The network virus propagation is influenced by various factors, and some of them are neglected in most of the existed models in the literature. In this paper, we study the network virus propagation based on the the epidemiological viewpoint. We assume that nodes can be equipped with protection against virus and the security of a node depends not only on his protection strategy but also by those chosen by other nodes in the network. A crucial aspect is whether owners of device, e.g., either smartphones, machines or tablets, are willing to be equipped to protect themselves or to take the risk to be contaminated in order to avoid the payment for a new antivirus. We model the interaction between nodes as a non-cooperative games where the node has two strategies: either to update the antivirus or not.
To this aim, we provide a full characterization of the equilibria of the game and we investigate the impact of the price of protection on the equilibrium as well as the efficiency of the protection at equilibrium. Further we consider more realistic scenarios in which the dynamic of sources that disseminate the virus, evolves as function of the popularity of virus. In this work, the interest in the virus by sources evolves under the Influence Linear Threshold (HILT) model.

\end{abstract}
\section{Introduction}\label{Introduction}

The Internet continues to grow exponentially and many applications continue to multiply on the Internet, with immediate benefits to end users. However, these network-based applications and services can pose security risks to devises. Recently many attacks have been launched against business, users and governments, that are attributed to some decentralized online communities acting anonymously in a coordinated manner. However, despite the important efforts spent by the many security companies, researchers, and government institutes, information systems’ security is still a great concern \cite{networksecurity},\cite{cloudsecurity}. One of important security risks is the propagation of some sophistical virus in the internet in which each infected node becomes a new source of infection. The problem of virus propagation has been studied through huge papers focusing on mainly on epidemic thresholds for real and realistic network and immunization policies \cite{immunization},\cite{Masuda},\cite{Omic}.
Many researchers have taken help of the biological system to study the behavior of spread of virus in a computer network and how to immune the computer system \cite{biology}. These epidemic models were very useful in network security modeling and immunization strategy.

To manage the network security, lots of efforts have been devoted to study virus propagation and characteristics \cite{propagationvirusNetworksecurity}.
To protect from the spread of virus in the network, nodes, e.g., either smartphones, machines or tablets, can use some anti-virus software with curing tools. In existing research, there are many perspectives to protect against network infections \cite{Infocomprotectiongame}. An important issue of protection is the frequency of the update in order to provide a protection against new virus. But many of the anti-virus applications are of server-client in which the system may provide the last update especially when there is an intensive epidemic of new virus. But an important fundamental source of difficulty in developing efficient protection is the difficulty to fully observe and control the network. As a consequence, full control and observability is impossible, leading to systems that are vulnerable to local as well as remote attacks. Other factor has to be considered when evaluating the security risk from virus, is the decentralized decision on the protection. Indeed, the protection against virus is typically autonomous nature of decision making in the network and the performance of such security is usually made under the assumption that nodes are willing to use an anti-virus for protection. But, a crucial aspect is whether owners of devices are willing to be equipped to protect them self or to take the risk to be contaminated in order to avoid the payment for a antivirus. Any successful security solution should consider those factors.

To date, the problem that attached the most attention from the research community is how efficiency the protection against the virus when the decisions are taken autonomously by nodes according to the cost of antivirus as well as the risk to be contaminated. In fact, the bigger the number of nodes equipped with protection, the lower the infection probability for a node without protection. We model the problem as non-cooperative games in order to establish the behavior of nodes against virus infection. Furthermore, we consider a source $S$ which propagates the virus through the network. We are concerned with spread of sources that disseminate the virus. In particular, we associate the dynamic of sources with the popularity of virus which is measured by the number of infected nodes at least one time. We model this influence process using
Homogeneous Influence Linear Threshold (HILT) model \cite{HILT}. The HILT model focuses on the threshold
behavior in influence propagation, which we can frequently
relate to — when enough of our friends bought a product, we may be influenced and converted to follow the same action. In our context, when a virus reaches some level of popularity, other sources may participate in dissemination of this virus. \\

The remainder of this paper is organized as follows. The epidemic model and the network game is introduced in section \ref{Network model}. In section \ref{Characterization of Equilibrium}, we study the dynamics and the different equilibrium properties of the proposed security game. Numerical illustrations of the system behaviors and the equilibrium characteristics are given in Section \ref{Numerical evaluation}. In order to enhance the readability of the paper, the mathematical proofs are given in Appendices.
\section{Network model}\label{Network model}

Consider a sparse network that consists of a large population of $N$ computer systems (CS) or nodes. We assume that the network is a complete graph where all computers can communicate with each others. A set of sources generates a virus with a rate $\mu_s$ and nodes in the network becomes susceptible to be infected by this virus. Each node $i$ chooses or not to be equipped in order to protect himself from the virus by paying a relative price of protection. Mixed strategies, i.e, probability distribution over the actions, are also possible. Obviously, all nodes have an incentive to protect themselves until the virus extinction. However, every antivirus update costs a price $U_c$. Hence, the strategy adopted by a node corresponds to a certain utility it receives and this utility depends on actions performed by $N$ nodes. Nodes with outdated antivirus are vulnerable to the virus spread process, and lose an infection cost $I_c$ if they were infected. An infected node can recover after a curing time using various tools (e.g., through a clean-up software). Under this setting, nodes shall immunise themselves during the period of the virus spread while minimizing the antivirus update cost.
\begin{figure}[t]

\hspace{2 cm}
\includegraphics[width=4cm,height=2.5cm]{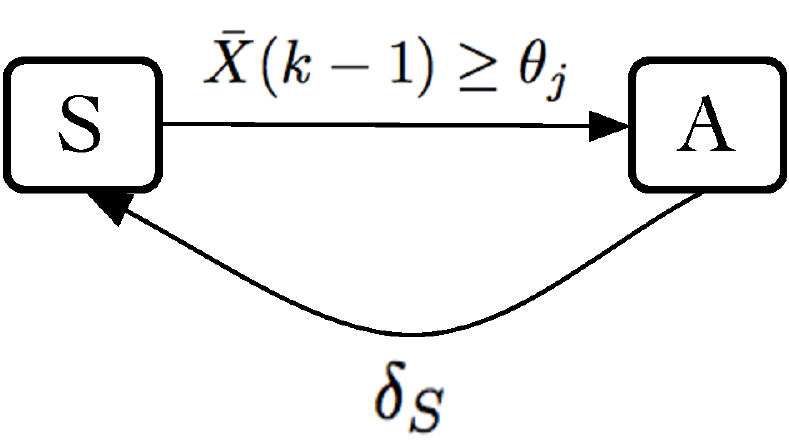}

\caption{Transition condition of activation state.} \label{fig:S2}
\end{figure}
 \subsection{Modeling active sources evolution}

 Assume that there are a set of sources that generate the virus. We associate the dynamic of active sources, which still disseminate the virus, with the dynamic of the number of infected nodes. In particular, we associate the dynamic of the number of active sources with the popularity of virus which is measured by the number of infected nodes at least one time. We model this influence process using Homogeneous Influence Linear Threshold (HILT) model \cite{HILT}. The HILT model focuses on the threshold behavior in influence propagation. The evolution of the number of active sources is modeled following process evolving in continuous time. Each source $j$ chooses a random threshold $\theta_j$ from an arbitrary threshold distribution with cumulative density function (c.d.f) $F$. Hence a source becomes active if the popularity, which measured by the number of infected nodes, exceeds its threshold. An active source disseminate the virus during a random time $T_m$, which follows the exponential distribution with rate $\delta_S$. Under the HILT model, the following proposition describes the dynamic of the number of actives sources at time $t$.
 \begin{figure}[t]

\hspace{2cm}
\includegraphics[width=4cm,height=3cm]{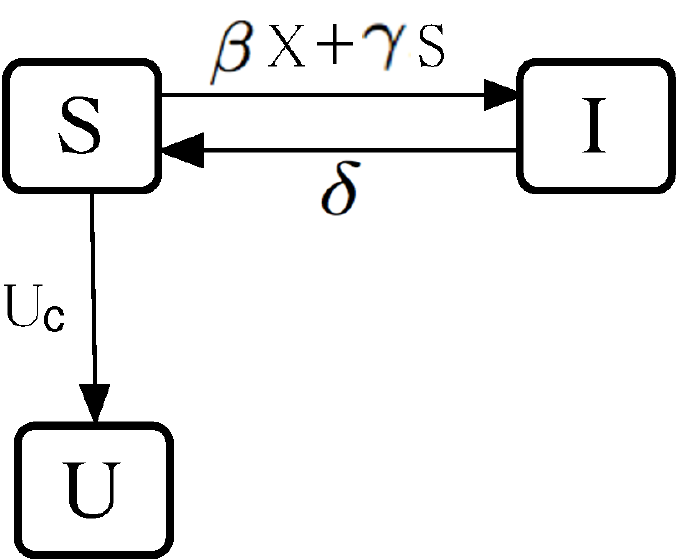}

\caption{Flow chart of states' transitions. } \label{fig:RD}
\end{figure}
\begin{proposition}
The dynamic of the number of active sources that disseminate the virus, is given by
$$\dot{S}(t)=- \delta_S S(t) + \lambda \frac{ \displaystyle f(\bar{X}(t))}{ \displaystyle 1-F(\bar{X}(t))} (N-S(t)),$$
where $\bar X(t)$ is the number of infected nodes till time $t$.
\end{proposition}

\begin{proof} See appendix.
\end{proof}
 \subsection{Modeling Infection Dynamic}

 To model the spreading process under the influence of a curing process, we choose the SIS (Susceptible Infected Susceptible) model, which is one of the most studied epidemic models \cite{SISVirusSpread,SISConductance}. A node at time $t$ can be in one out of two states: infected or susceptible. We assume that the curing process is a Poisson process with rate $\delta$, after which the node becomes healthy, but susceptible again to the virus. Let $X(t)$ denote the number of infected nodes and each infected node can infect healthy nodes with a contact rate $\beta$. Note that the total infection rate of a susceptible node is $\beta$ times the size of infected nodes.

 \subsection{The Security Game}
 We consider a security game, in which nodes choose individually whether or not to invest in the protection by updating their antivirus versions. Indeed, nodes prefer to not invest in antivirus update if the network is enough protected, i.e., there is enough nodes in the network that have chosen to update. Each node has two strategies: either to invest in the antivirus protection, i.e., pure strategy \emph{update} ($U$), or not to invest, i.e., pure strategy \emph{not update} ($NU$). Each strategy corresponds to certain playoff for the node. Notice that, a node may also be protected by the other nodes' update: the risk to be infected decreases with the number of antivirus activated throughout the network. Accordingly, the payoff of a node depends on the actions performed by the $N-1$ nodes. \\
 We denote by $V_j(k_U)$ the long term fitness of a computer, given that it plays the strategy $j$ $\in \{U, NU\}$, and that $k_U$ is the number of updated antivirus.
 The fitness is given by
\begin{equation}\label{eq:fitness}
 V_j(k_U) = \left\{\begin{array}{lr}
- U_c & j \in U \\
- P_i I_c & j \notin U
\end{array}
\right.\\
\end{equation}
where $P_i$ is the probability to be infected until the virus extinction.
 \section{Characterization of Equilibrium}\label{Characterization of Equilibrium}

The nodes which invest in the antivirus protection are directly immune to the virus. Therefore, they can not infect the other nodes or be infected. The dynamic of sources $S(t)$ depends on the the virus popularity. The activation process describes how the infectious nodes cause the inactive sources to become active.
 A source $j$ is influenced by the accumulative number of infected nodes $\bar{X}$ and it will be active when
 \begin{equation}\label{eq:Activationcondition}
 \bar{X}(t) \ge \theta
 \end{equation}
 We assume that a node contacts an active source with a rate $\gamma$. Before evaluating the dynamics of the the infected nodes $X^N(t)$, we study the dynamic of sources $S^N(t)$ under the activation process.\\
 A source $S$ is active when the number of infected nodes $\bar{X}^N(t)$ reaches the target value $\theta$. Let $\bar{X}(t)$ be the dynamic of infected nodes disregarding the curing process. The dynamics of $\bar{X}(t)$ is given by
 $$\dot{\bar{X}}(t)=(\beta X(t)+ \gamma S(t)) (N-k_U-X(t))$$
 \begin{proof} See appendix.
\end{proof}
 Recall that $S^N (t)$ is the set of active sources which participate in the infection process by time $t$, and $\delta_S S^N (t)$ is the set of the sources which are no longer interested to the virus and move from the active state to the susceptible state. A source is influenced by the accumulative infection process with a rate $\lambda$. By applying Condition (\ref{eq:Activationcondition}), we can write the sources dynamics as follows
   $$\dot{S}(t)=- \delta_S S(t) + \lambda \frac{ \displaystyle f(\bar{X}(t))}{ \displaystyle 1-F(\bar{X}(t))} (N-S(t))$$

 We know that the hazard function \cite{cox1962renewal} for the c.d.f $F(.)$ is given by $h_F(x)=\frac{f(x)}{1-F(x)}$ where $f(x)$ is the corresponding density function. Hence, the ordinary differential equation (ODE) becomes\\

  $$\dot{S}(t)=- \delta_S S(t) + \lambda h_F(\bar{X}(t)) (N-S(t))$$
  \begin{proof} See appendix.
\end{proof}
 The sources contact $(N-k_U-X^N(t))$ susceptible nodes with a rate $\gamma$.Therefore, we can write the dynamics of $X(t)$ as following

 $$\dot{X}(t) = - \delta X(t) + ( \beta X(t)+ \gamma S(t)) (N-k_U-X(t))$$

 This growth equation gives the dynamics of infected nodes under the sources activation process.
 All nodes aim to be enough protected during the lifetime of the virus. A node $i$ can be infected by a source or by another infected node. The infection probability $P_{i}$ is expressed as follows
 \begin{equation}\label{eq:infectionprobability}
 \,\,\,\,\,P_i(t)= 1 -\E\left[ \int_0^{\tau_c} e^{- (\beta x(t)+\gamma S(t))} dt\right]
 \end{equation}
 \subsubsection{Pure Nash Equilibrium}
 \begin{definition}
 At a Nash equilibrium (NE), no player can improve its fitness by unilaterally deviating from the equilibrium.\\
 \end{definition}
 For the proposed game a NE in pure strategies exists
if and only if the following two conditions are satisfied
 \begin{equation}\label{NEconditions}
\forall 1\le j \le N = \left\{\begin{array}{lr}
V_j(NU,k_U-1) \le \,\,\, V_j(U,k_U) \\
V_j(NU,k_U) \,\,\,\,\,\,\,\,\,\, \ge \,\,\, V_j(U,k_U+1)
\end{array}
\right.
\end{equation}
 We are interested in the existence and uniqueness of the pure NE which is characterized by the number $\Psi$ of players investing in the antivirus.\\
 A unique pure NE exists for the proposed security game when $V_j(NU, \psi)=V_j(U, \psi)$.
 \begin{proof} See appendix.
\end{proof}
 \subsubsection{Mixed Nash Equilibrium}

Let us now discuss the case when every node maintain a probability distribution over the two actions $(U,NU)$. The advantage of this mixed equilibrium compared to the pure one is that a node can invest in protection only for a fraction of the time and stay susceptible the rest of the time. This kind of equilibrium is more efficient for our case because we study a homogeneous population with fixed update and infection cost.
 In a mixed strategy game, a node $i$ can decide to invest in protection (playing $U$) with probability $p_i$ or keep protected only by his neighbors (playing $NU$) with probability $(1-p_i)$.\\
 $\textbf{p}=(p_1,p_2,\ldots,p_N), \forall i p_i \ge 0$, is the mixed strategy profile. For $p_i \notin \{0,1\}$ we have a fully mixed strategy profile. We note $(p_i,p_{-i})$ if node $i$ uses strategy $p_i$ and other use $p_{-i}=(p_1,\ldots,p_{i-1},p_{i+1},\ldots,p_N)$.\\
 We denote by $V_j(p,p_{-i})$ the playoff of a node $i$ which invest in antivirus with probability $p$.
\begin{definition}
A mixed strategy $p{^*_i} \in [0,1]$ is a NE if for each player $i$ (where $i=1,\ldots,N$) we have \begin{equation}\label{MNE}
U_i(p^*_1,\ldots,p^*_{i-1},p^*_i,p^*_{i+1},.,p^*_N)\ge U_i(p^*_1,\ldots,p^*_{i-1},p_i,p^*_{i+1},\ldots,p^*_N)
\end{equation}for every mixed strategy $p_i \in [0,1]$.\\
If $\forall i, p^*{_i} \notin \{0,1\}$ then we call $p^*$ fully mixed NE.
\end{definition}
Every finite strategic game has a mixed strategy NE\cite{nashequilibrium}.
There exists a unique fully mixed NE $p^*$ for the proposed game and it is solution of
\begin{equation}\label{eq:fullmixNE}
 \sum_{k=1}^{N} C^{N-1}_{k-1}(p^*)^{k-1} (1-p^*)^{N-k} V(U,k)=0
\end{equation}
\begin{proof} See appendix.
\end{proof}
 \subsubsection{Equilibrium with Mixers and Non-Mixers}
 The mixers are the players that choose a mixed strategy. We suppose that a part of the population chooses to play a pure strategy $U$ or $NU$ and the rest of the players are mixers. We will study the existence of the equilibrium in this case. Let $N_U \in \{0,1,\ldots,N\}$ be the number of players choosing the pure strategy $U$, and $N_{NU} \in \{0,1,\ldots,N\}$ be the number of players choosing the pure strategy $NU$.\\
 The $N-N_U-N_{NU}$ players use the mixed strategy. Let $p^* \in (0,1)$ be the probability with which the mixers choose the strategy $U$. Moreover, we denote by $V_U(N_U,N_{NU},p)$ the fitness of the node who updates its antivirus and $V_{NU}(N_U,N_{NU},p)$ the fitness for the node who does not update its antivirus.
 A necessary condition for the strategy $(N_U,N_{NU}, p^*)$ to be a NE (with at list one mixer) is that the mixer is indifferent whether it chooses a pure strategy $U$ or $NU$. This translates mathematically as follows\\
 \begin{equation}\label{eq:MNEcondition}
V_U(N_U+1,N_{NU},p^*) = V_U(N_U,N_{NU}+1,p^*)
\end{equation}
 \begin{proof} See appendix.
\end{proof}

 A unique NE of type $(N_U, N_{NU},p^*)$ exist for this case, and is solution of
\begin{multline}
\label{eq:mixernonmixerNE}
\sum_{k=0}^{N-N_U-N_{NU}} C^{N-N_U-N_{NU}}_{k-1} (p^*)^{k-1}  \\
\qquad \qquad \cdot (1-p^*)^{N-N_U-N_{NU}-k} V(U,N_U+k)=0.
\end{multline}

 We prove that this NE of type $(N_U, N_{NU},p^*)$ exists only for $N_U<\psi$ and $N_U+N_{NU} \le N-2$.
   \begin{proof} See appendix.
\end{proof}
  In this section, we have studied different NE types under the $S(t)$ activation process.
  We summarize the different Nash equilibrium types as following:
  \begin{itemize}
    \item \textbf{Pure Nash Equilibrium}: There exists a unique NE when the utility of $U$ is equal to the utility of $NU$ and we must update exactly $\psi$ nodes to get this equilibrium,
    \item \textbf{Mixed Nash Equilibrium}: A unique fully mixed NE $p^*$ exists and it is solution of Equation (\ref{eq:fullmixNE}),
   \item \textbf{Mixer and Non-Mixer Nash Equilibrium}: We characterize this equilibrium by a necessary condition (\ref{eq:MNEcondition}). A unique NE exists and it is solution of Equation (\ref{eq:mixernonmixerNE}). Moreover, we have proved that a NE exists only for $N_U < \psi$ and for $N_U+N_{NU} \le N-2$.
 \end{itemize}
   \section{Numerical evaluation
}\label{Numerical evaluation}
 In this section, we provide a numerical analysis of the performances of the proposed security game. We first evaluate the infection probability at the equilibrium. To do so we solve Equation (\ref{eq:fullmixNE}) to get the activation probability at the equilibrium. We show how the activation and the infection process depend on the system parameters, such as the number of nodes $N$ and the update cost $Uc$.
 \subsection{System characteristics}
 Fig. \ref{fig:Xevolution} and Fig. \ref{fig:Sevolution} illustrate the behavior of the infected nodes $X(t)$ and the sources $S(t)$ as function of the time for different activation probabilities $(0.01, 0.1, 0.5)$ and a contact rate $\beta=1\times 10^{-3}$. We further take $X(0)=0$ and $S(0)=5$. As expected, we remark a cause-effect phenomenon between the nodes and sources. The number of infected nodes increases as a result of the virus spread till reaching a given infection rate. Then, when the virus popularity reaches a certain level, the participating in the virus spread increases yielding an increase in the number of sources.\\
  In the proposed game model, activation probability is a fundamental parameter and is related to how many nodes install the new anti-virus software. To analyze effects of the activation probability, $p$ is set to three different values $(0.01, 0.1, 0.5)$.
 Fig. \ref{fig:Xevolution} shows that the number of infected nodes $X(t)$ slightly fluctuates before reaching a stable (absorbing) state. In general, the higher the activation probability is, the faster $X(t)$ decreases. This is due to the fact that increasing the activation probability implies a decrease in the risk of being infected for susceptible nodes. Notice that, depending on the activation probability, the virus may disappear completely or become scars. We will discuss this point later in the paper.\\
 \begin{figure}[t]

\hspace{0cm}
\includegraphics[width=9cm,height=5cm]{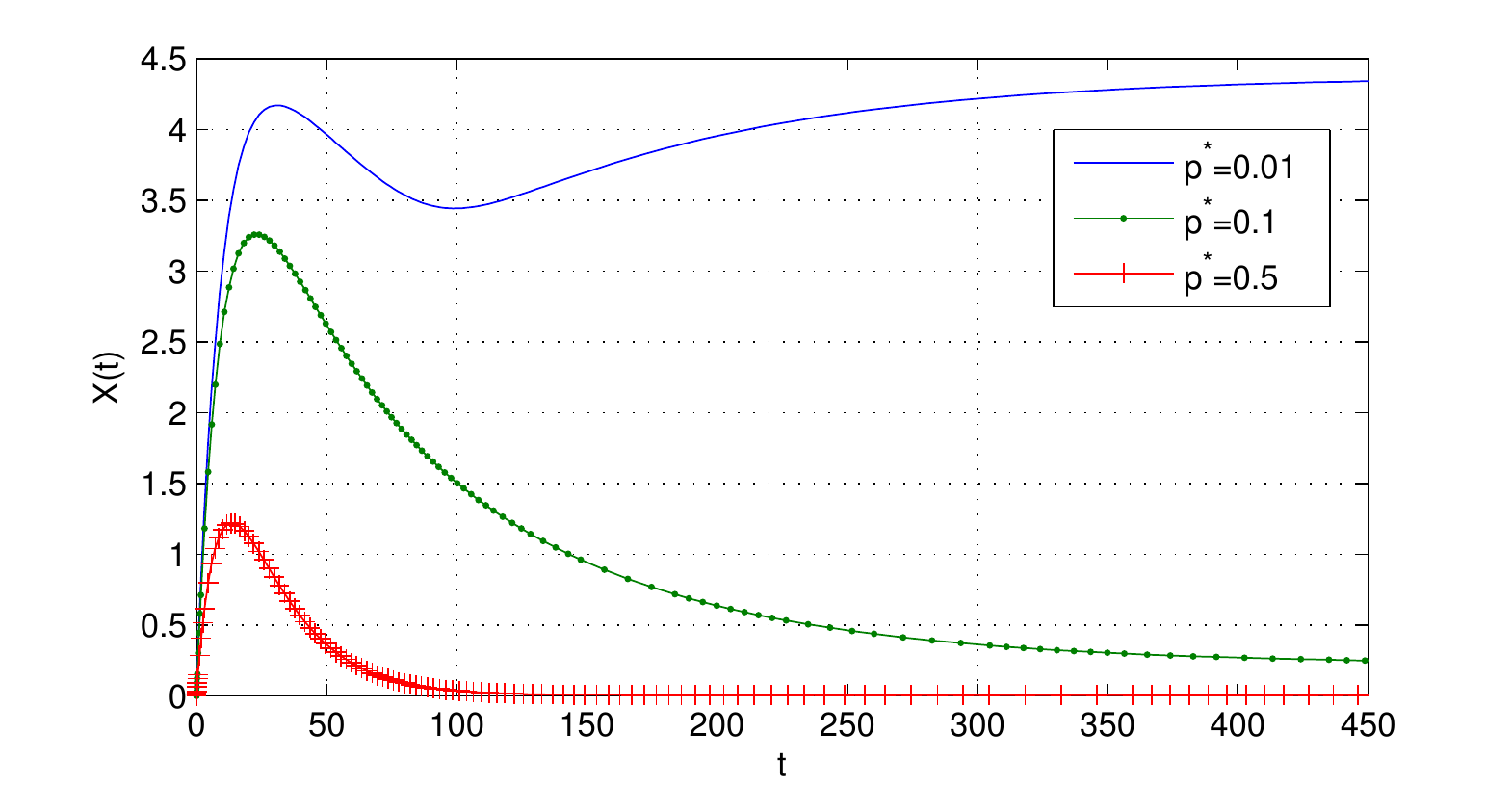}

\caption{The infection process for different activation probabilities $p$, where $N=100$, $N_s=50$, $\beta=1 \times10^{-3}$, $\gamma=1 \times10^{-3}$, $\delta=1 \times10^{-1}$, $\delta_S=1 \times10^{-1}$, $\lambda=5\times10^{-6}$, $X(0)=0$, $S(0)=5$, $Ic=1$ and $Uc=0.1$.} \label{fig:Xevolution}
\end{figure}
 \begin{figure}[t]

\hspace{0cm}
\includegraphics[width=9cm,height=5cm]{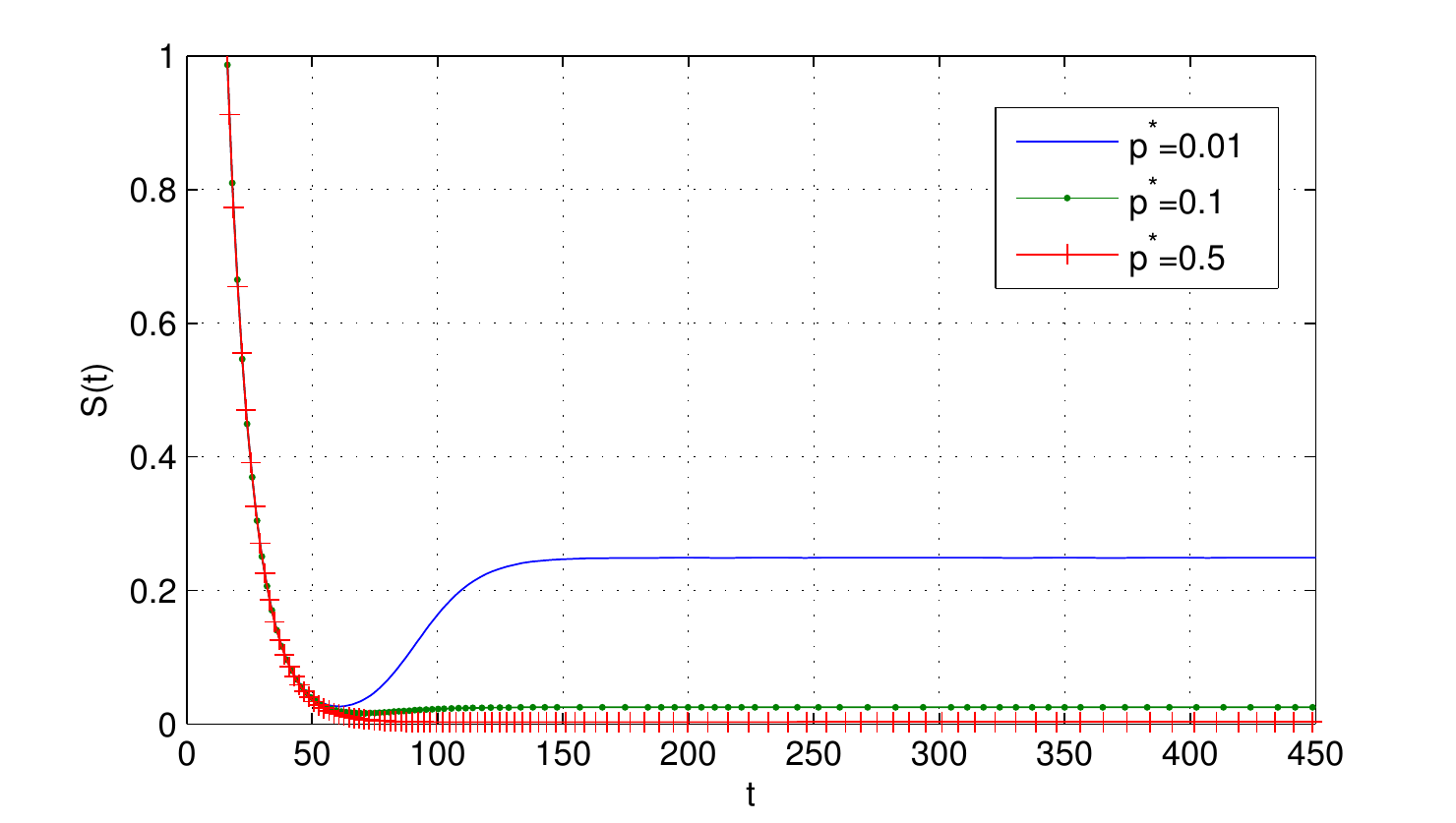}

\caption{Sources behaviour for different activation probabilities.} \label{fig:Sevolution}
\end{figure}
 Fig. \ref{fig:Sevolution} depicts the dynamics of the interested sources in the virus spread for different $p$. We clearly notice that, for low activation probability values, e.g., $p=0.01$, $S(t)$ decreases until the virus popularity reaches a target value. When the activation probability increases to $p=0.5$, we can see that the number of sources are decreasing gradually to vanish eventually.\\

 \begin{figure}[t]

\hspace{0cm}
\includegraphics[width=9cm,height=5cm]{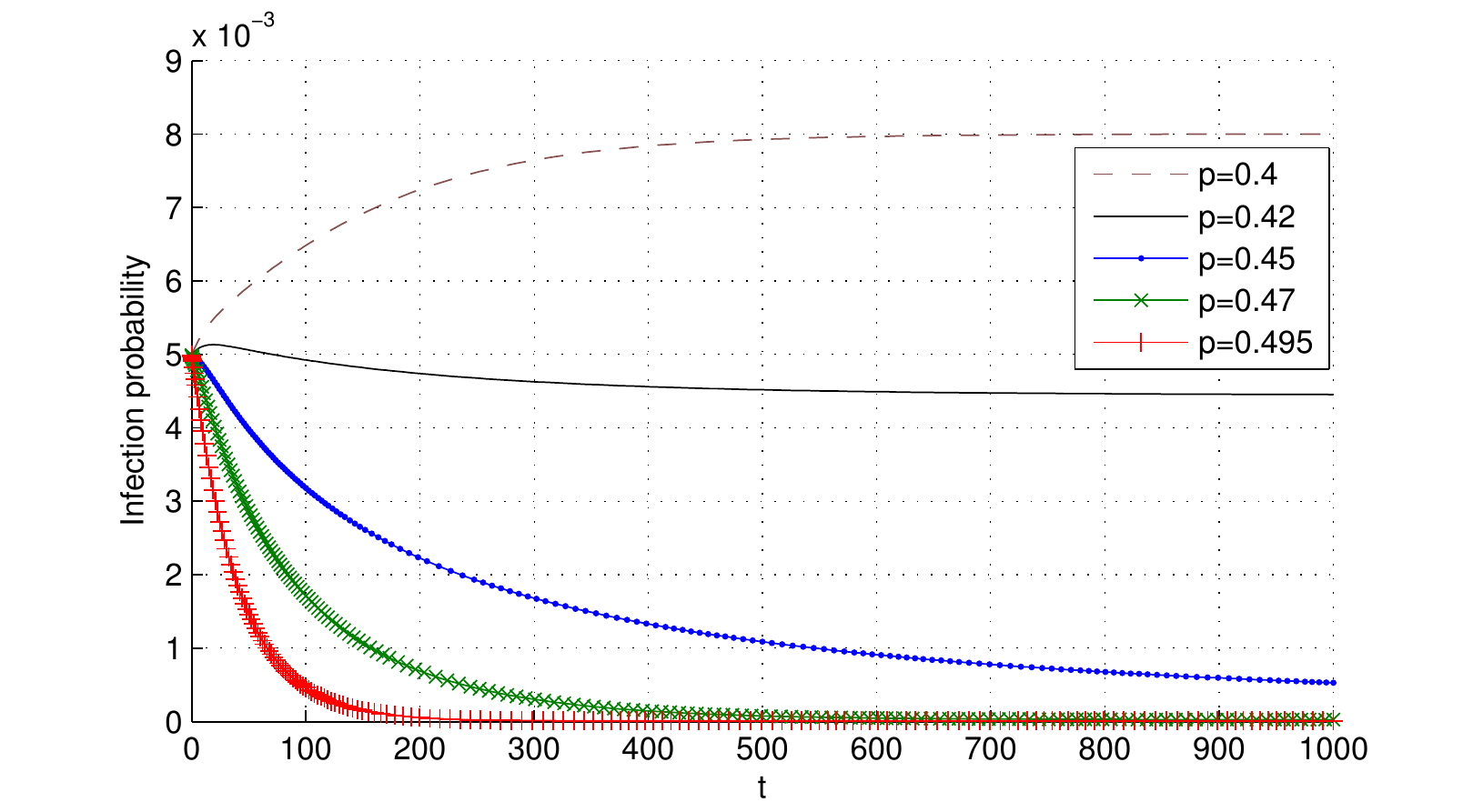}

\caption{The infection probability for different activation probabilities, where $N=100$, $\beta=1 \times10^{-3}$, $\gamma=1 \times10^{-3}$, $\delta=1 \times10^{-1}$, $\delta_S=1 \times10^{-1}$, $\lambda=1\times10^{-4}$, $X(0)=0$, $S(0)=5$, $Ic=1$ and $Uc=0.1$.} \label{fig:infectionprobabilitychangingp}
\end{figure}

Fig. \ref{fig:infectionprobabilitychangingp} illustrates the time evolution of the infection probability for different activation probabilities $p$. We remark that, from $p= 0.495$, the infection probability monotonically decreases till completely vanishing at $t=230$. This suggests that using an activation probability higher that $0.495$ is worthless as, from $p= 0.495$, the virus is going to disappear in any case.
\\
 Unless otherwise stated, we will use the following parameters:
$N=500$, $N_s=50$, $\beta=1 \times 10^{-4}$, $\gamma=1 \times 10^{-3}$, $\delta=0.1$, $\delta_S=0.1$, $\lambda=1 \times 10^{-4}$, $X(0)=0$, $S(0)=10$, $Ic=1$ and $Uc=0.1$. We notice by $t_f$ the time of epidemic extinction corresponding to the time for which we have $X(t)=0$. For this parameters, a virus extinction time $t_f$ exists and we can compute the infection probability in $[0,\ldots,t_f]$.
 \subsection{System Performances at the Equilibrium }
  \subsubsection{Fully Mixed equilibrium }
 Let us now evaluate the performances of the proposed security game. We characterize the equilibrium in Section \ref{Characterization of Equilibrium} by solving the polynomial Equation in (\ref{eq:fullmixNE}). We solve (\ref{eq:fullmixNE}) to get a unique solution $p^*=0.29$ at the equilibrium.
 \subsubsection{Mixer and Non-Mixer Equilibrium }
 For this case, we have $N_U$ nodes that always update their antivirus (pure strategy update) and $N_{NU}$ that never update their antivirus (pure strategy not update). The $N-N_{NU}-N_U$ mixers update the antivirus with an activation probability $p$. Verifying Condition (\ref{eq:MNEcondition}), we vary the $N_U$ and the $N_{NU}$ to find the activation probability for the mixers at the equilibrium. For $N=500$ and $Uc=0.1$, we get $p^*=0.19$ for all the mixers when $N_U=50$ and $N_{NU}=70$.
 \subsection{Effects of network size}
 To proceed further with the analysis, we resort to evaluate the impact of the network size on the system behavior at the equilibrium.
\begin{figure}[t]

\hspace{0cm}
\includegraphics[width=9cm,height=5cm]{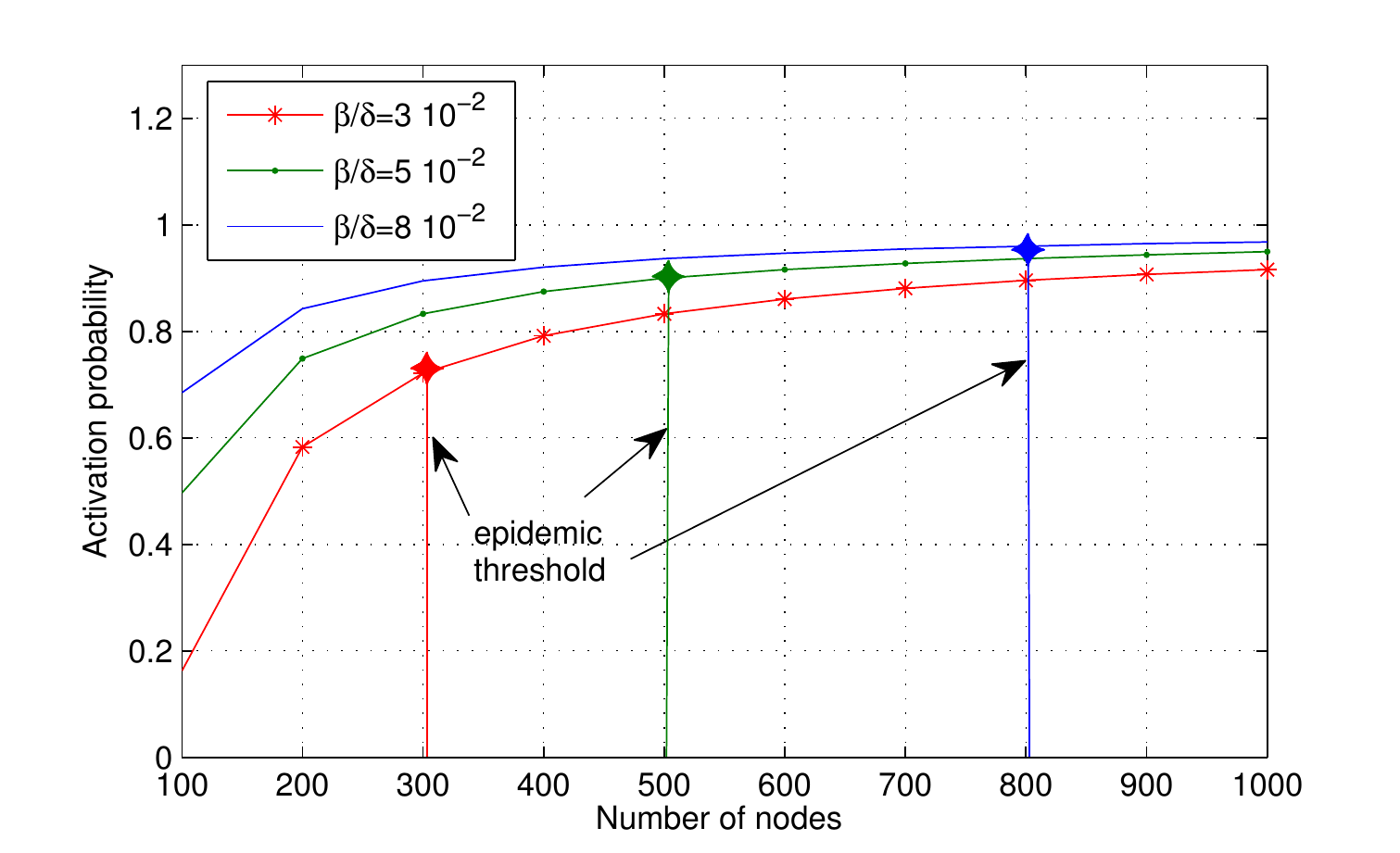}

\caption{The activation probability for increasing number of nodes.} \label{fig:activationN}
\end{figure}

In Fig. \ref{fig:activationN}, we plot the activation probability as function of the network size. It is clearly shown that the activation probability increases when the network is larger which is somehow obvious as a larger network tends to provide greater risk of infection among nodes.
On the other hand, we know from \cite{EpidemicThreshold} that the relation between the epidemic threshold $\tau_c$ and the transmission to disease-induced mortality ratio  $\frac{\beta}{\delta}$ translates to the following conditions:\\
 \bi
 \item if $\frac{\beta}{\delta}<\tau_c$, the virus dies out over time,
\item if $\frac{\beta}{\delta}>\tau_c$, the virus survives and the infection becomes an epidemic.
 \ei
 In the particular case of a complete graph, the epidemic threshold is given by $\tau_c=\frac{1}{N-1}$. Thus, in the proposed network, the infection dies when $N< \frac{\delta}{\beta}+1$. Increasing $\frac{\beta}{\delta}$, the number of nodes satisfying the epidemic threshold is larger. This result gives incentive to companies to manage their network parameters ($\beta$ ,$\delta$, $N$) so that the infection dies out over time. \\
 To evaluate the performance of the equilibrium, we compare, in Fig. \ref{fig:gain}, the gain $G=\frac{U_c N - p^* U_c N}{U_c N}$ to simpler policies where one may activate a given percentage of the nodes in the system depending on the general policy of the company. For instance, a conservative policy is more likely to update $90\%$ of the nodes, whereas a lax policy is more likely to update $10\%$ of the entire nodes in the system.\\
 So far, we have been  interested in the influence of network parameters ($\beta$, $\delta$, $N$) on the company server's security management. Now, we will study how the network parameters influence the antivirus producers decision.

\begin{figure}[t]

\hspace{0cm}
\includegraphics[width=9cm,height=5cm]{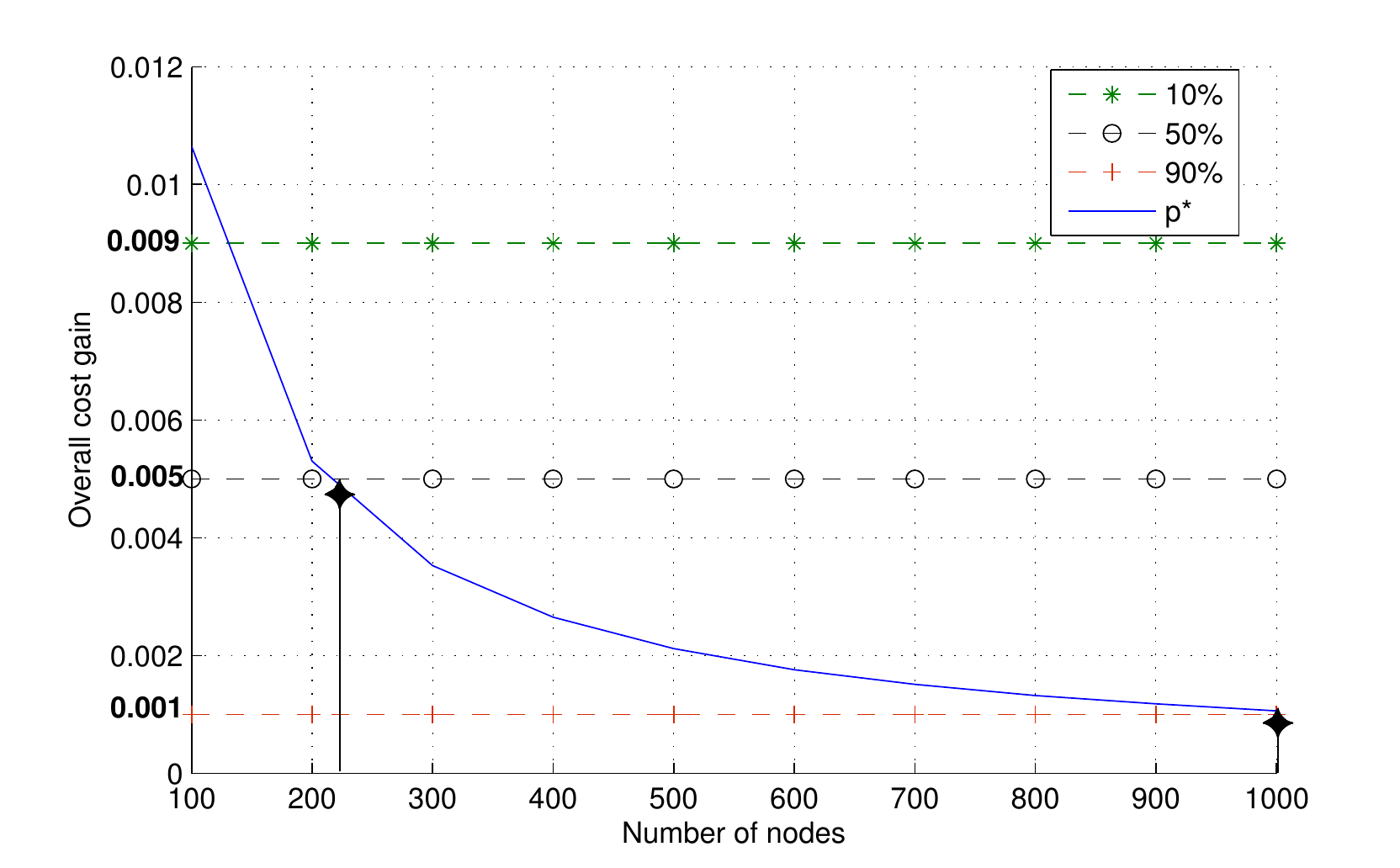}

\caption{The overall cost gain as function of the number of nodes.} \label{fig:gain}
\end{figure}
 \subsection{Effects of update cost}
 \begin{figure}[t]

\hspace{0.3cm}
\includegraphics[width=8cm,height=5cm]{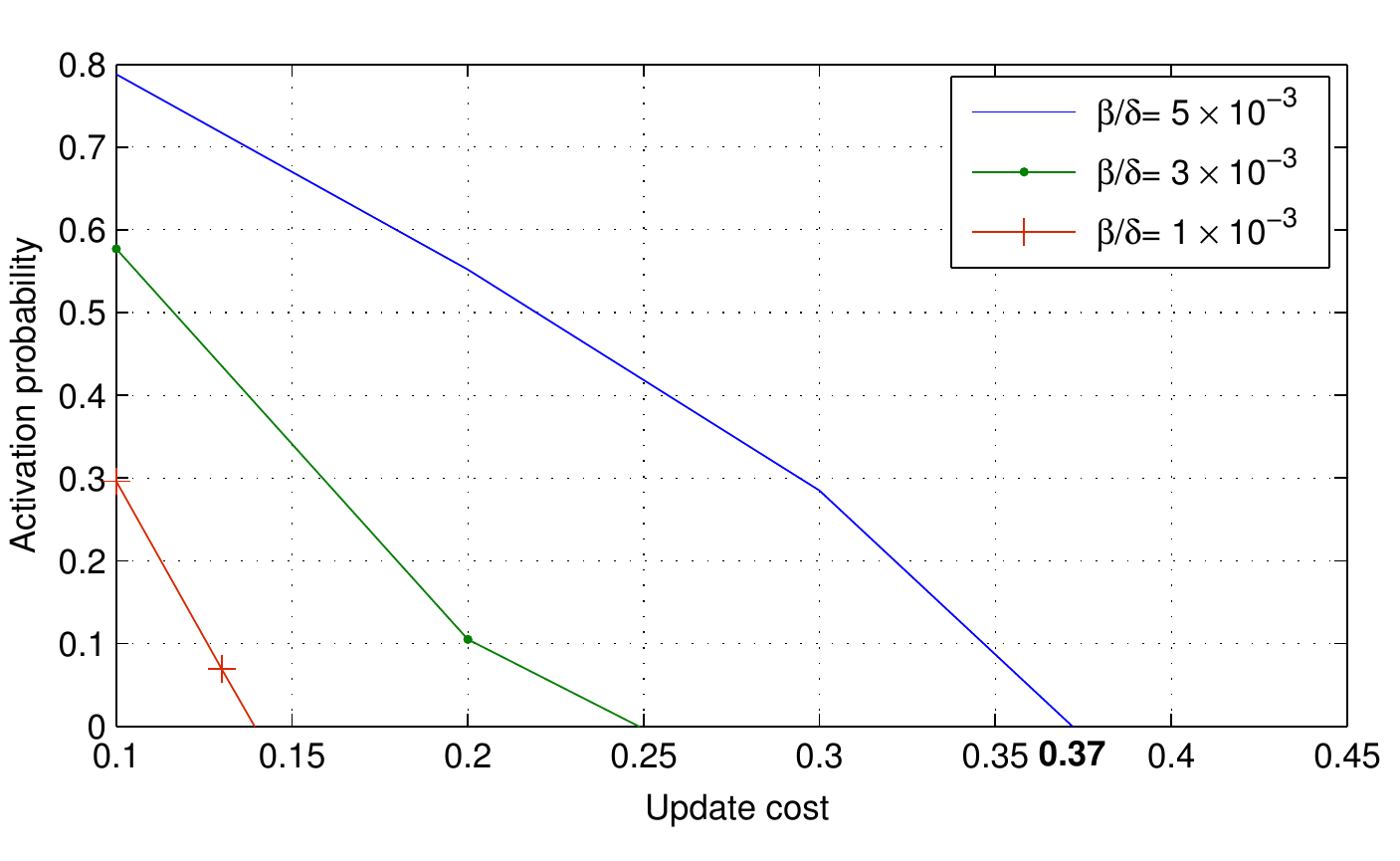}

\caption{The activation probability as function of the update cost. } \label{fig:activationCost}
\end{figure}
Let us now study the impact of the update cost on the system behavior at the equilibrium.\\
Fig. \ref{fig:activationCost} gives the time evolution of the activation probability considering different values of $\beta/\delta$. We find the update cost $U_c^*$ for which the probability of activation is equal to $0$. This specific value $U_c^*$ is very important for the antivirus producers to manage the $U_c$, as approaching $U_c$ the this limit value $U_c^*$. The $U_c^*$ increases with $\beta/\delta$, as the risk of infection increases with $\beta/\delta$.

\begin{figure}[t]

\hspace{0cm}
\includegraphics[width=9.3cm,height=5cm]{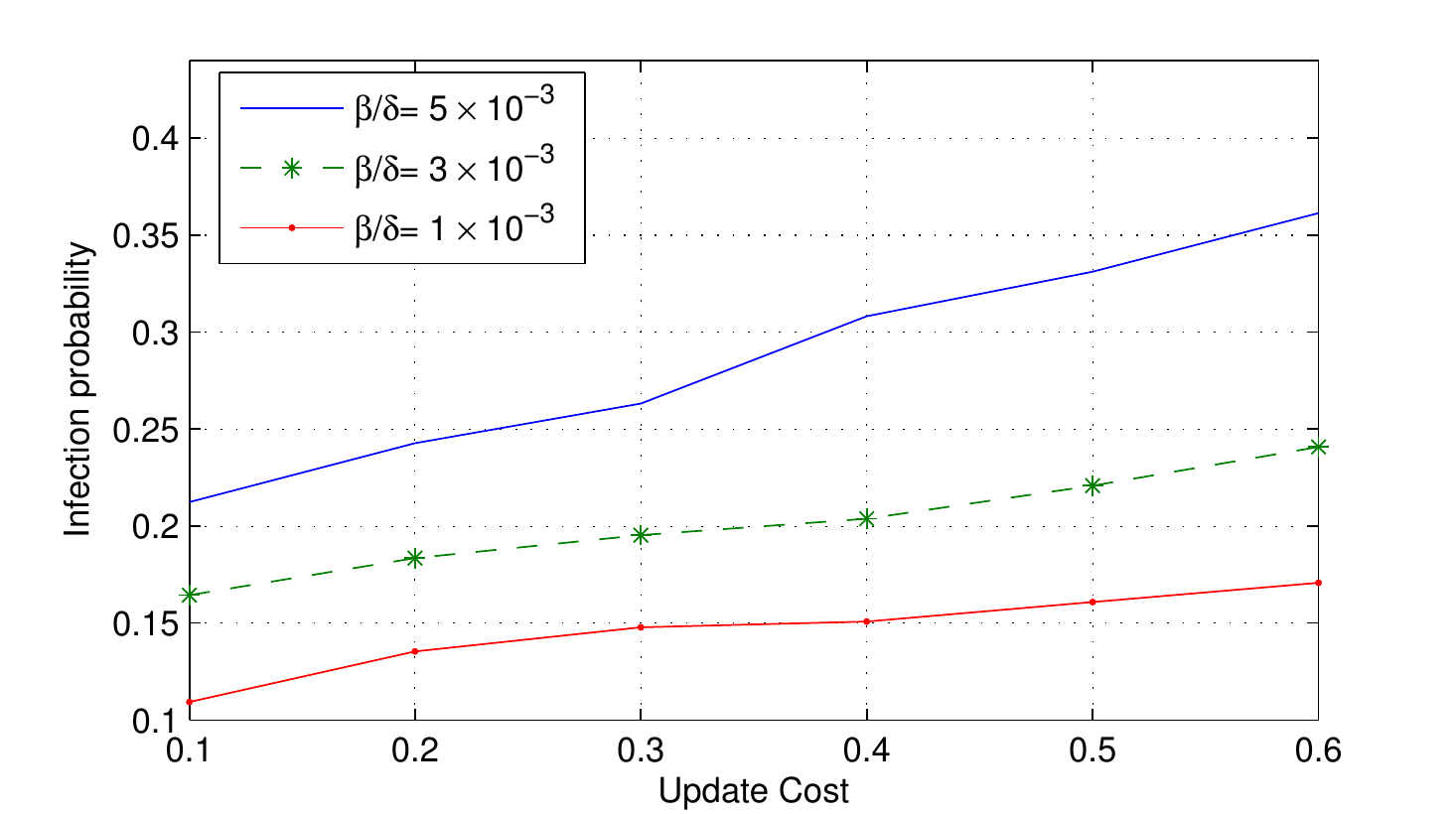}

\caption{The infection probability as function of the update cost.} \label{fig:infectionCost}
\end{figure}
Fig. \ref{fig:infectionCost} illustrates the infection probability in $[0,\ldots,t_f]$ as function of the update cost. At the equilibrium, the infection probability increases with the update cost. This is justified by the fact that the number of nodes participating in the antivirus activation decreases as the update cost increases, yielding a higher infection risk.
 \begin{figure}[t]

\hspace{0cm}
\includegraphics[width=9cm,height=5cm]{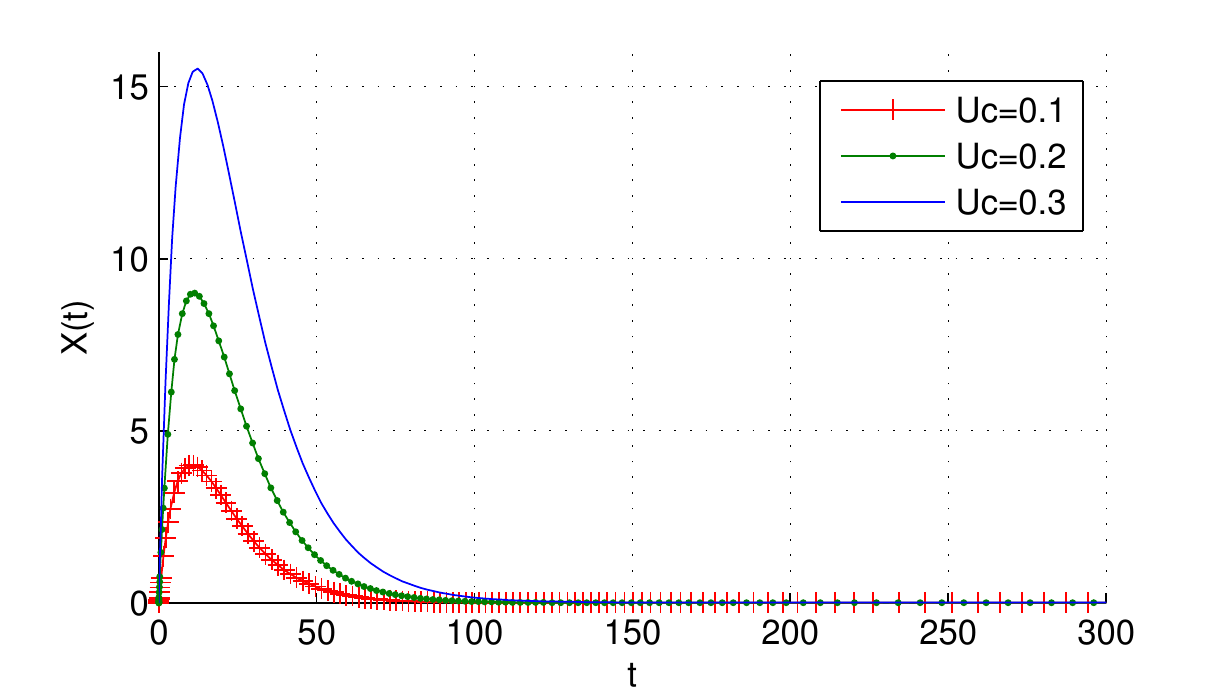}

\caption{Infected nodes for different update costs with $N=100$, $N_s=50$, $\beta=1 \times 10^{-4}$, $\gamma=1 \times 10^{-3}$, $\delta=0.1$, $\delta_S=0.1$, $\lambda=1 \times 10^{-4}$, $X(0)=0$, $S(0)=10$ and $Ic=1$.} \label{fig:XvariationCost}
\end{figure}
In Fig. \ref{fig:XvariationCost}, we plot the infected node evolution in time at the equilibrium considering different update costs $U_c$. We observe that $X(t)$ increases until a certain $t^*$ and then it decreases till the epidemic extinction, i.e., when $X(t)=0$. It is shown that the more the update cost is lower, the more the $t^*$ and the epidemic extinction are lower. This result helps the antivirus producers to manage the update cost. Increasing the $U_c$, we have more infected nodes in the network, but reaching $U_c^*$, no node is interested in the update. The antivirus producers are interested to reach a target performance of update cost to maximize the gain.

\section{CONCLUSION
}\label{CONCLUSION}

We have studied a game theoretic model for network virus protection under an activation process. The virus spread dynamics is modelled as an epidemic process. The strategy of each node is either to update or not to update the antivirus. We have established the existence and the uniqueness of different types of Nash equilibria. We have studied the dynamic of both the infection process and the activation process. Both the company server's security management and the antivirus producer's strategies have been addressed. It has been shown that, depending on the network topology, the companies have incentive to manage their network parameters in order to ensure that the infection dies out finite-time horizons.
\section*{Appendix}
\subsection{Dynamic of infected nodes}
 $X^N(t+\Delta t) = X^N(t) ( 1 - \delta \Delta t) \\+ (\beta^N X^N(t)+ \gamma^N S^N (t)) (N-U-X^N(t)) \Delta t$
 \\
 \\
 $ \frac { \displaystyle X^N(t+\Delta t) - X^N(t)}{ \displaystyle \Delta t}= - \delta X^N(t) \\+( \beta^N X^N(t)+ \gamma^N S^N (t)) (N-U-X^N(t))$\\

  We evaluate this expression for $\Delta t$ $\rightarrow$ 0 to get:\\

 $\dot{X}(t) = - \delta X(t) + ( \beta X(t)+ \gamma S(t)) (N-U-X(t))$\\

 $\bar{X}^N(t+\Delta t) = \bar{X}^N(t) \\+ (\beta^N X^N(t)+ \gamma^N S^N(t)) (N-U-X^N(t)) \Delta t$\\

  $ \frac { \displaystyle \bar{X}^N(t+\Delta t) - \bar{X}^N(t)}{ \displaystyle \Delta t}= ( \beta^N X^N(t)+ \gamma^N S^N (t)) (N-U-X^N(t))$\\
 We evaluate this expression for $\Delta t$ $\rightarrow$ 0 to get:\\
 $\dot{\bar{X}}(t)=(\beta X(t)+ \gamma S(t)) (N-U-X(t)) $\\
\subsection{Dynamic of the active sources}
 $S^N (t+\Delta t) = S^N (t) ( 1 - \delta_S \Delta t) \\+ \lambda \frac{\displaystyle F(\bar{X}^N(t+\Delta t))-F(\bar{X}^N(t))}{ \displaystyle1-F(\bar{X}^N(t))} (N_H-S^N (t))$\\
 \\
 \\
 $\frac {\displaystyle S^N (t+\Delta t)-S^N (t)}{ \displaystyle \Delta t}= -\delta_S S^N (t) \\+ \lambda \frac{ \frac{ \displaystyle F(\bar{X}^N(t+\Delta t))-F(\bar{X}^N(t))}{ \displaystyle \Delta t}}{ \displaystyle 1-F(\bar{X}^N(t))} (N_H-S^N (t))$\\
 Let $f(\cdot)$ be the density function of $F(\cdot)$, we have\\
 $\lim\limits_{\Delta t \to 0}$ $\frac{ \displaystyle F(\bar{X}^N(t+\Delta t))-F(\bar{X}^N(t))}{ \displaystyle \Delta t}$=$f(\bar{X}(t))$\\
 We evaluate this expression for $\Delta t \rightarrow 0$:\\
 $\dot{S}(t)=- \delta_S S(t) + \lambda \frac{ \displaystyle f(\bar{X}(t))}{ \displaystyle 1-F(\bar{X}(t))} (N-S(t))$\\

 We know that the hazard function \cite{cox1962renewal} for the c.d.f $F(.)$ is given by $h_F(x)=\frac{f(x)}{1-F(x)}$ where $f(x)$ is the corresponding density function. Hence, the ODE can be written as following,\\

  $\dot{S}(t)=- \delta_S S(t) + \lambda h_F(\bar{X}(t)) (N-S(t))$\\

\subsection{Pure Nash equilibrium}
 $P_{i}(k_U)$ is decreasing in the number of updates $K_U$:\\
  $P_{inf}( \psi-1) \ge P_{inf}( \psi) $\\
  $I_c > 0$ and fixed for all players so:\\
  $-I_c P_{inf}( \psi-1) \le-I_c P_{inf}( \psi)$\\
  $-I_c P_{inf}( \psi-1) \le V_j(NU, \psi)=V_j(U, \psi)$\\
  $V_j(NU, \psi-1) \le V_j(U, \psi)$\\
  which is the first Condition (\ref{NEconditions}.1).\\
  $V_j(NU, \psi)=V_j(U, \psi)=-U_c=V_j(U, \psi+1)$\\
  and we have the second Condition (\ref{NEconditions}.2).\\
     We conclude that $V_j(NU, \psi)=V_j(U, \psi)$) is a pure NE. Hence, a NE at $\psi$ exists for the proposed game.\\

   Second, we will prove by contradiction the uniqueness of $\psi$ at the NE.\\
 Let $k_U > \psi$, we have: \\ $ V_j(NU, k_U-1 ) > V_j(U, k_U-1 ) = -U_c = V_j(U, k_U)$ \\
 $ V_j(NU, k_U ) > V_j(U, k_U)$, so that (\ref{NEconditions}.1) fails. \\
Let $K_u < \psi$, we have: \\ $ V_j(NU, k_U+1 ) < V_j(U, k_U+1 ) = -U_c = V_j(U, k_U)$\\
$ V_j(NU, k_U+1 ) < V_j(U, k_U)$, so that (\ref{NEconditions}.2) fails. \\
 The existence and the uniqueness of NE means that the proposed security game has only one stable point at $\psi$ when $V_j(NU, \psi) = V_j(U,\psi)$.
\subsection{Mixed Nash equilibrium}
Let $p$ the symmetric mixed strategy adopted by every node in the proposed game, $p_i=p$, $\forall i$.
The new fitness of every node will be given by:\\
$V_j(p^*_i,p_{-i})= p_i^* \sum_{k=1}^{N} C^{N-1}_{k-1} p_{-i}^{k-1} (1-p_{-i})^{N-k} V(U,k) + (1-p_i^*) \sum_{i=1}^{N} C^{N-1}_{k-1} p_{-i}^{k-1} (1-p_{-i})^{N-k} V(NU,k)$\\
We have: $V(NU,k)$=$- V(U,k)$ so,\\
$V_j(p^*_i,p_{-i})$= $(2p^*_i - 1) \sum_{k=1}^{N} C^{N-1}_{k-1} p_{-i}^{k-1} (1-p_{-i})^{N-k} V(U,k)$\\
=$ (2p^*_i - 1) D(N,p_i) $\\
with $D(N,p_i) = \sum_{k=1}^{N} C^{N-1}_{k-1} p_{-i}^{k-1} (1-p_{-i})^{N-k} V(U,k)$\\
A mixed strategy is obtained when $D(N,p_i) = 0$ because for $D(N,p_i) < 0$ the best response for the node $i$ is $p_i=0$ and $p_i=1$ is a best response when $D(N,p_i) > 0$.
We characterize the mixed equilibrium by
\begin{equation}\label{eq:equilibriumcharactersticpolynome}
D(N,p_i) = \sum_{k=1}^{N} C^{N-1}_{k-1} p_{-i}^{k-1} (1-p_{-i})^{N-k} V(U,k)=0
\end{equation}
We have $D(N,0)= V(U,1)$ and $D(N,1)=V(U,N)$. From (\ref{eq:infectionprobability}), the infection probability is decreasing with the number of active antivirus.
 the fitness of updated system $U$ is decreasing with the number $k_U$ of updates. We get $V(U,N)<V(U,1)$,
$D(N,0)> D(N,1)$.Thus, $D(N,p_i)$ is strictly decreasing with $p$. Hence, the mixed symmetric NE is unique.

\subsection{Mixers and no-Mixers Nash Equilibrium}
\subsubsection{Equilibrium condition}
To prove the necessary condition to have a NE in our case, we have to show that under this condition the Non-mixes can't improve their fitness with diving from the equilibrium.
 We have $V_U(N_U,N_{NU},p)$ is decreasing with $N_U$.\\
 $V_U(N_U,N_{NU},p^*) \ge V_U(N_U+1,N_{NU},p^*)$ and from (\ref{eq:MNEcondition})\\
 $V_U(N_U+1,N_{NU},p^*)=V_{NU}(N_U,N_{NU}+1,p^*)$ \\
  $V_U(N_U,N_{NU},p^*) \ge V_{NU}(N_U,N_{NU}+1,p^*)$\\
 $\ge V_{NU}(N_U-1,N_{NU}+1,p^*)$\\
 $\ge p^*V_U(N_U,N_{NU},p^*) + (1-p^*)V_{NU}(N_U-1,N_{NU}+1,p^*)$\\
 Therefore, under Condition (\ref{eq:MNEcondition}) a player who plays $U$ strategy cannot improve his fitness by deviating from the strategy $(N_U,N_{NU},p^*)$ using any strategy $p^* \in [0,1[$.\\
\subsubsection{Equilibrium}
 We have proved that $k_U$=$\psi$ is a pure NE. \\
 Suppose that there is at least one player updates its antivirus. Let $p=0^+$, reps. $p=1^-$ be the mixed strategy infinitely close to zero, reps. to one. We have $V_U(N_U,N_{NU},p^*)$ is strictly decreasing with $N_U$ and $p^*$. We can write (\ref{eq:MNEcondition}) as follows\\
 $$V_{U}(N_U+1,N_{NU},0^+) > V_{NU}(N_U,N_{NU}+1,0^+)$$
  $$V_{U}(N_U+1,N_{NU},1^-) \le V_{NU}(N_U,N_{NU}+1,1^-)$$
 If $N_U \ge \psi$, then\\
 $V_{U}(N_U+1,N_{NU},0^+) \le V_{NU}(N_U,N_{NU}+1,0^+)$ wich does not satisfy Condition (\ref{eq:MNEcondition}) to have a NE.\\
  A NE exists for $l < \psi$.\\
 For $N_U+N_{NU}+1 > N-1$, there is also no NE.\\
 We suppose that we have $l < \psi$, let $N_U+N_{NU}+1= N$ (only one node plays mixed strategy) then $\forall p$,\\
 $V_{U}(N_U+1,N_{NU},p)= C_1$ and
  $V_{NU}(N_U,N_{NU}+1,p)= C_2$ \\
   Moreover $V_{U}(N_U,N_{NU},p)$ is strictly decreasing with $N_U$, we have $C_2 > C_1$ which contradicts Condition (\ref{eq:MNEcondition}).\\
   We conclude that the NE of type $(N_U, N_{NU},p^*)$ exists only for $N_U<\psi$ and $N_U+N_{NU} \le N-2$.
  \bibliographystyle{ieeetr}
 \bibliography{sis_bib}
 \end{document}